\documentclass[11pt,a4paper]{article}
\usepackage[utf8]{inputenc}
\usepackage{amsmath}
\usepackage{amsfonts}
\usepackage{amssymb}
\usepackage{amsthm,epsfig,epstopdf,titling,url,array,float}
\usepackage{slashbox}
\usepackage{mathtools}
\usepackage{geometry}
\usepackage{setspace}
\usepackage{graphicx}
\usepackage{authblk}
\theoremstyle{plain}
\newtheorem{thm}{Theorem}[section]
\newtheorem{lem}[thm]{Lemma}

\theoremstyle{definition}

\theoremstyle{remark}

\let\given\givenbase

\title{Optimal Gamma density to Obfuscate Quantitative data with Added Noise}
\author[a]{Debolina Ghatak }
\affil[a]{TCG Centres for Research and Education
in Science and Technology}
\author[b]{Debasis Sengupta}
\affil[b]{Indian Statistical Institute Kolkata}
\author[c]{Bimal Roy}
\affil[c]{Indian Statistical Institute Kolkata}
\date{}
\begin{document}
\maketitle
\doublespacing

\begin{abstract}
Protecting the privacy of individuals in a data-set is no less important than making statistical inferences from it. In case the data in hand is quantitative, the usual way to protect it is to add a noise to the individual data values. But, what should be an ideal density used to generate the noise, so that we can get the maximum use of the data, without compromising privacy? In this paper, we deal with this problem and propose a method of selecting a density within the Gamma family that is optimal for this purpose.
\end{abstract}
\section{Introduction}

The use of data in making inferences and future predictions has been a popular pursuit for many decades. But, with the advance of various internet activities and due to increasing use of e-data, privacy protection has become more essential than ever. In order to protect the data from intruders, sometimes agencies perturb or mask the data in some way so that the intruder cannot guess the original individual data values, but anybody can get an idea about the underlying statistics like mean, median, variance etc.

In case the data-values corresponding to the sensitive attribute in hand are discrete, one may apply the Post-Randomization method as introduced by Gouweleew et al. (1998)\cite{JKWD} and later discussed by Nayak et al.(2011a,2015,2016)\cite{TBL}\cite{NAS}\cite{TAZ}, Robello-Monedaro (2010)\cite{DJJ}, Andrieu et al. (2003)\cite{CNAM}, Press et al. (2007)\cite{WSWB}, Mares and Shlomo (2014)\cite{MS}. If the sensitive attribute in hand is quantitative, one may either swap the data according to the methods discussed by Dalenius and Reiss (1982)\cite{TDSPR}, Moore (1996)\cite{MRA}, Murlidhar et al. (1999)\cite{KRR}, Sarathy et al. (2002)\cite{RK} or generate synthetic data, i.e., data from the distribution of the true data-set as discussed by Rubin (1993)\cite{DBR}, Reiter and Kinney (2003)\cite{JRSK}, or may simply put (either add or multiply) a noise to the true data-values when the noise is generated independently of the original data-set. While in the former methods, i.e., data swapping or generating synthetic data, one can treat the perturbed data as the original data in making inferences, but in those cases the correlation information of the sensitive attribute with other attributes gets erased, which is not desired. On the other hand, if one puts noise to it, the correlation information is retained, but the distribution changes. However, if the noise distribution is known, along with the parameters involved, then the distribution curve can be estimated from the masked values and masking distribution, provided the error is chosen from a suitable distribution.

Mathematically, let $\{X_1,X_2, \cdots, X_n\}$ be the true data-set which is assumed to be independently distributed and for each $i \in \{1,2, \cdots, n\}$, and $X_i$ follows a certain distribution with unknown cumulative distribution function (CDF) $G(\cdot)$ and density $g(\cdot)$. Let the obfuscating noise $\{Y_i,i \in \{1,2, \cdots, n\}\}$ be a sample from a known distribution with CDF $F(\cdot)$ ( and density $f(\cdot)$) independent of $X_i$. Suppose $Z_i=X_i+Y_i$ instead of $X_i$. $\{Z_i,i \in \{1,2, \cdots, n\}\}$ is the masked data-set which comes from the unknown distribution with CDF $H(\cdot)$. Since $X_i$ and $Y_i$ are independent and continuous random variables, $Z_i$ also has a continuous CDF $H(\cdot)$, which is the convolution of $F(\cdot)$ and $G(\cdot)$. This model is known as the Additive Noise Model.

An interesting problem concerning the additive noise model is the choice of an optimal noise density function $f(\cdot)$ that ensures privacy protection and is good for estimation of $g(\cdot)$. Dealing with this problem as a whole is hard, and the solution would depend on the estimator. An estimator with well-understood large sample properties is the deconvolution kernel density estimator proposed by Carroll and Hall (1988)\cite{CH88}:
\begin{equation}
\hat{g}(x)=\frac{1}{2\pi}\int_{-\infty}^\infty\exp(-itx)\tilde{K}(tb) \frac1n\sum_{i=1}^n\frac{\exp(itZ_i)}{\tilde{f}(t)}dt,\quad -\infty<x<\infty,
\label{Eqn:DKDE}
\end{equation}
where $\tilde{f}$ and $\tilde{K}$ are Fourier transforms of $f$ and a kernel function $K$, respectively, and $b$ is a bandwidth parameter. It is known that if $f(\cdot)$ is a density function whose Fourier transform $\tilde{f}(\cdot)$ satisfies
\begin{equation}
c_1(1+|t|)^{-\tilde{\vartheta}} \leq |\tilde{f}(t)| \leq c_2(1+|t|)^{-\tilde{\vartheta}}  \mbox{, $-\infty<t<\infty$}
\label{Eqn:OrdSmth}
\end{equation}
for some positive constants $c_2>c_1$ and $\tilde{\vartheta}>0$, then the estimator \eqref{Eqn:DKDE}, using a suitably chosen bandwidth $b$, have mean squared error (MSE) with an algebraic rate of convergence (i.e., the MSE converges to 0 with the rate $n^{-\iota}$ for some $\iota>0$; see Meister (2005)~\cite{AM}). Densities satisfying the condition~\eqref{Eqn:OrdSmth} are called ordinary smooth density functions. The rate of convergence happens to be slower for some other noise densities such as Normal and Cauchy. The Laplace distribution is known to have an ordinary smooth density function. The following lemma shows that members of the two sided Gamma Family having density $f_{\vartheta,\eta}(\cdot)$ given by
\begin{equation}
f_{\vartheta,\eta}(x)=\frac{1}{2 \Gamma(\vartheta) \eta^\vartheta}|x|^{\vartheta-1}e^{-\frac{|x|}{\eta}} \mbox{ , for $-\infty<x<\infty$}
\label{Eqn:GammaFamily}
\end{equation}
and indexed by the scale parameter $\eta>0$ and shape parameter $\vartheta>0$ also has ordinary smooth density when $\vartheta \le1$. Note that the special case $\vartheta=1$ corresponds to the Laplace distribution.

\begin{lem}
\label{whygamma}
The density  $f_{\vartheta,\eta}(\cdot)$ given by~\eqref{Eqn:GammaFamily} satisfies the condition \eqref{Eqn:OrdSmth} for some $c_1,c_2>0$ when $\vartheta \leq 1$, but does not satisfy it for any $c_1,c_2>0$ when $\vartheta>1$.
\end{lem}

\begin{proof} See Appendix.\end{proof}

Lemma~\ref{whygamma} suggests that the two-sided gamma family of distributions would be a restricted but reasonable class to search for a suitable noise distribution for obfuscation. Once a search criterion is fixed, one has only to do a parametric optimization rather than a functional one.


In the next section we point out some shortcomings of using the measures of confidentiality given by Tendick (1991) \cite{TDK}, Spruill (1983) \cite{NCS} and Ghatak and Roy (2018)\cite{GR} in the present problem, and provide a new measure. It turns out that the Laplace distribution is not necessarily the optimal choice. In Section 3 we discuss the procedure for estimating the distribution curve from convoluted data. In Section 4, we discuss our approach of finding the optimal density function to obfuscate a given data-set within the Gamma family. In Section 5, we give some simulation results to illustrate our work and finally end with some concluding remarks in Section 6.


\section{Measures of confidentiality}
\subsection{Problems with existing measures}

At first, given any shape parameter $\vartheta$, we want to find a scale parameter such that the additive noise distribution~\eqref{Eqn:GammaFamily} having this pair of parameters can protect the data sufficiently. However, to do so, we would need a measure of disclosure risk to the data. Tendick (1991)~\cite{TDK} proposed a measure of confidentiality measure when true data and noise are both normal. The paper suggests the use of $\rho^2=\frac{\sigma^2}{\sigma^2+\sigma_N^2}$ as a measure of confidentiality, where $X \sim N(\mu,\sigma^2)$, $Y \sim N(0,\sigma_N^2)$. If the value of $\rho^2$ is large, the data are less obfuscated. Since the normal density is not ordinary smooth, a measure defined for the normal case would be of limited utility. In any case, this gross measure give little indication of the worst-case level of confidentiality. Spruill et al.(1983)~\cite{NCS} suggested the measure of confidentiality given by one minus the fraction of cases where the nearest match of the obfuscated value among the original data-set occurs in the current position. This measure is clearly empirical. If one uses its expectation, that would be a complicated function of the distribution of $X$ and $Y$. Depending on the noise density, the probability of best match occurring in the correct place might be rather small. The measure of confidentiality that has earned the most attention in the last decade is the idea of differential privacy given by Dwork et. al. (2009)\cite{CS}. There have been some work on statistical mechanisms that ensure differential privacy including the smoothed histogram mechanism and orthogonal series density estimation \cite{WZ}\cite{RH}. But these mechanisms encounter a lot of loss in utility of data making it of no use to the statistician.

\begin{figure}
\centering
\includegraphics[scale=0.6]{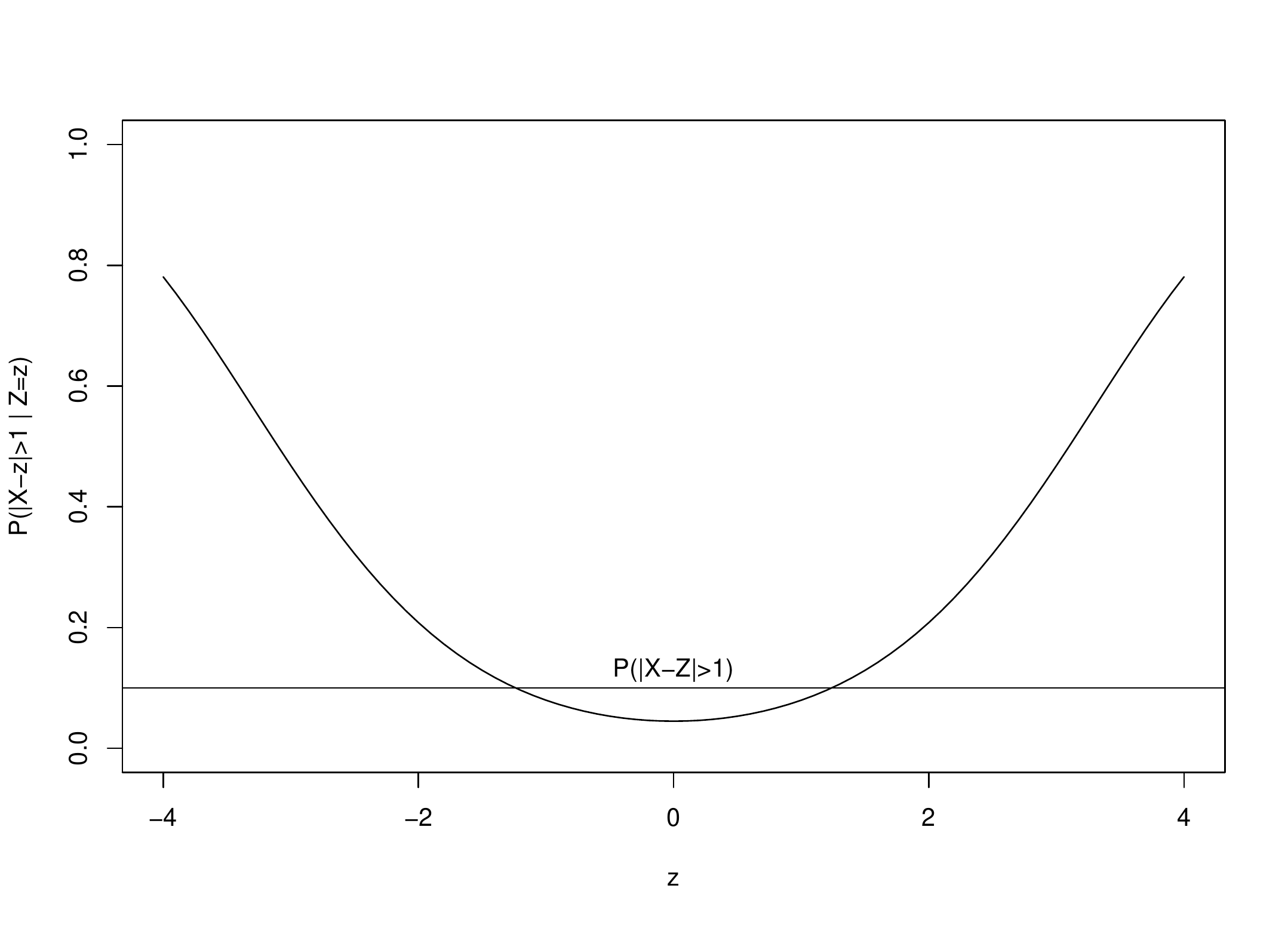}
\caption{Figure showing $P(|X-z|>\epsilon\given Z=z)$ for different $z$, when $X$ is standard normal and $Y$ is Laplace with such a scale parameter that $P(|Y|>\epsilon)=\delta$ for $\epsilon=1$ and $\delta=0.1$}
\label{Fig:alpha}
\end{figure}

Ghatak and Roy (2018)\cite{GR} proposed that the obfuscating distribution is chosen to make sure that
\begin{equation}
P[|Y|<\epsilon]=1-\delta,
\label{Eqn:FixedQuantile}
\end{equation}
where $Y$ has density given by \eqref{Eqn:GammaFamily} and $\delta$, $\epsilon$ are suitably chosen. This method can be used to find a scale parameter $\eta$ which, for any given $\vartheta$, would produce obfuscating noise above a specified threshold $\epsilon$ with specified probability $\delta$. However, obfuscation is fundamentally a matter of limiting the chances of anyone guessing accurately the value of $X$ from $Z$, rather than ensuring a certain amount of noise. In order to illustrate the distinction, let us consider standard normal $X$ and Laplace $Y$. If we set $\epsilon=1$ and $\delta=0.1$, then the scale parameter of $Y$ has to be $\eta=0.4342945$. Obfuscation would be good if for any given value of $Z$, the actual value $X$ is at least $\epsilon$ away from $Z$ with probability at least $\delta$. The plot of the conditional probability $P(|X-z|>\epsilon\given Z=z)$ against $z$ is shown in Figure~\ref{Fig:alpha}. The unconditional probability $P(|X-Z|>\epsilon)=\delta$ is shown as a horizontal straight line, for reference. It transpires that even though the unconditional probability is 0.1, the conditional probability can be much smaller than 0.1 when $z$ is small. In particular, the conditional probability is only 0.045 for $z=0$. Thus, smaller absolute values of $Z$ entail much weaker level of obfuscation than what the criterion \eqref{Eqn:FixedQuantile} suggests. 

This problem motivated us to think of a better measure of confidentiality for the problem at hand. Frank (1978)~\cite{FNK} and Paass (1988)~\cite{GP} have discussed the importance of the conditional distribution of the true data density given the convoluted data (here, the conditional distribution of $X$ given $Z$) in determining the uncertainty apparent to the intruder. Our new measure will be based on this conditional distribution.

\subsection{A new measure}

Suppose the additive noise distribution satisfies the following assumptions.
\paragraph{Assumptions (A1)}
\begin{itemize}
\item[(i)] The noise density is an even function, i.e., $f(-y)=f(y)$ for all $y \in \mathbb{R}$.
\item[(ii)] $E[Y]$ exists (in which case it is equal to zero).
\item[(iii)] $0<Var(Y)<\infty$.
\end{itemize}

The Minimum Mean Squared Error (MMSE) predictor of $X$ in terms of $Z$ is $E[X \given Z]$. This is not necessarily equal to $Z$. Nevertheless, $Z$ is a simple and unbiased predictor of $X$ under the assumption A1(ii), and it does not require any knowledge of the distribution of $X$ (Better predictors may be constructed on the basis of that knowledge; e.g., if $Z=1.05$ and the range of the true data is known to be $[0,1]$, then the constant~1 is a better predictor than $Z$.) However, it is to be noted that it is reasonable to look for a measure of confidentiality that considers the deviation of $X$ from the observed value of~$Z$.

Starting with the assumptions A1, and using the notation $\sigma_X^2$ for the variance of $X$, consider the probability of $X$ lying in $\epsilon\sigma_X$ boundary of its predicted value $z$:
\begin{equation}
M_{z,\epsilon}^{(f)} =P[|X-Z|<\epsilon\sigma_X \given Z=z] =\dfrac{\int_{-\epsilon\sigma_X }^{\epsilon\sigma_X}{g(z-x)f(x)dx}}{\int_{-\infty}^{\infty}{g(z-x)f(x)dx}},
\label{Eqn:M}
\end{equation}
for a chosen threshold $\epsilon>0$. A high value of this probability for small $\epsilon$ signifies a high risk of disclosure. For fixed $f$ and $z$, the function $M_{z,\epsilon}^{(f)}$ is a continuous and non-decreasing function of $\epsilon$, with
$$\lim_{\epsilon\rightarrow0}M_{z,\epsilon}^{(f)}=0,\quad\mbox \quad\lim_{\epsilon\rightarrow\infty}M_{z,\epsilon}^{(f)}=1.$$
Thus, for each fixed $f$ and $z$ and $0<\delta<1$, one may find an $\epsilon^{\star}$ such that $M_{z,\epsilon^{\star}}^{(f)}=\delta$. For any fixed $\delta$ we would want this $\epsilon^{\star}$ to be as large as possible. Assuming obfuscation is needed for all values of $z \in \mathbb{R}$ (i.e., both low and high values need to be protected) a measure of comparison may be proposed as
$$ \mu^{(f,\delta)}=\min \{\epsilon>0 : \sup_{z \in \mathbb{R}}M_{z,\epsilon}^{(f)} \geq \delta \}.$$
This measure is the smallest multiplier $\epsilon$ such that $P[|X-Z|>\epsilon\sigma_X \given Z=z]\ge1-\delta$ for all $z\in \mathbb{R}$. Noise densities $f$ producing larger $\mu^{(f,\delta)}$ are more suitable for obfuscation. 

As an example of the computations involved, let us consider the case $X \sim N(0,\sigma_X^2)$ and $Y \sim N(0,\sigma_Y^2)$.

\begin{lem}
(Normal-Normal Obfuscation.) If $X \sim N(0,\sigma_X^2)$ and $Y \sim N(0,\sigma_Y^2)$ then for fixed $\delta$ 
$$\mu^{(f,\delta)}
=\tau_{\frac{1+\delta}{2}}\left/\sqrt{1+\frac{\sigma_X^2}{\sigma_Y^2}}\right.,$$
where $\tau_x$ is the  $x^{th}$ quantile of a standard normal variable.
\label{Res:NNobfs}
\end{lem}
\begin{proof} See Appendix.\end{proof}
It transpires that the new measure of confidentiality is a monotone function of Tendick's (1991)~\cite{TDK} measure defined for normal-normal obfuscation. The measure $\mu^{(f,\delta)}$ goes down to 0 as $\sigma_Y/\sigma_X$ goes to~0 (no obfuscation), while it goes up to $\tau_{\frac{1+\delta}{2}}$ as $\sigma_Y/\sigma_X$ goes to infinity (maximum obfuscation).

\subsection{An empirical measure of confidentiality}

While there is a closed form expression of $\mu^{(f,\delta)}$ in the normal-normal case, the computation may be more difficult when $Y$ has the gamma distribution and $X$ has some other specified distribution. In some situations, the distribution of $X$ may not even be known. For this reason, an empirical version of $\mu^{(f,\delta)}$ is now worked out.

One can rewrite \eqref{Eqn:M} as
$$M_{z,\epsilon}^{(f)} =\dfrac{\int_{z-\epsilon\sigma_X }^{z+\epsilon\sigma_X}{f(z-x)g(x)dx}}{\int_{-\infty}^{\infty}{f(z-x)g(x)dx}} =\dfrac{E[f(z-X)\mathbb{I}_{[z-\epsilon\sigma_X < X < z+\epsilon\sigma_X]}]}{E[f(z-X)]},
$$
where $\mathbb{I}_{B}$ denotes the indicator of the event $B$. In order to obtain an empirical version of this measure, one can replace the expectations in the numerator and denominator by empirical averages, and replace $\sigma_X$ by the square root of $\hat{\sigma}_X^2=\frac{1}{n-1}\sum_{j=1}^n{(X_j-\bar{X})^2}$, where $\bar{X}=\sum_{j=1}^n{X_j}$. This substitution leads us to the measure
\begin{equation}
\hat{\mu}_{X_1,X_2,\cdots,X_n}^{(f,\delta)}=\min \{ \epsilon>0 : \sup_{z \in \mathbb{R}}\hat{M}^{(f)}_{(z,\epsilon)} \geq \delta\}
\label{Eqn:Meas}
\end{equation}
where,
\begin{equation}
\hat{M}^{(f)}_{(z,\epsilon)} =\frac{\sum_{i=1}^n{f(z-X_i)\mathbb{I}_{[z-\epsilon\hat{\sigma}_X \leq X_i \leq z+ \epsilon\hat{\sigma}_X]}}}{\sum_{i=1}^n{f(z-X_i)}}.
\label{Eqn:Meas_stat}
\end{equation}
Computational problems may be avoided by limiting the supremum to the set 
$$\left\{z \in \mathbb{R}: \sum_{i=1}^n{f(z-X_i)}>0\right\}.$$

\section{Estimating density from obfuscated data}

To estimate the density function of original data, we use the deconvolution kernel density estimator \eqref{Eqn:DKDE} for reasons indicated in Section~1. This entails selection of the kernel and the bandwidth. For this purpose, we follow the procedure of Delaigle and Gijbels (2004)~\cite{ADIG}.
The chosen Kernel function has a Fourier transform given by, $\tilde{K}(t)=(1-t^2)^3\mathbb{I}_{[-1,1]}$. An exact expression of $K(\cdot)$ is available in Fan (1992)~\cite{FAN}. The bandwidth is chosen by minimizing the asymptotic integrated mean squared error (AIMSE), which is an approximation of the integrated mean squared error, i.e., $\int{E(\hat{g}(y)-g(y))^2dy}$, or, $\int{Var(\hat{g}(y))dy}$. The approximation, derived by Stefanski and Carroll (1990)~\cite{STC}, can be expressed as
$$ AIMSE(b)=\frac{1}{2\pi n}\int_{-\infty}^{\infty}{\frac{|\tilde{K}(tb)|^2}{|\tilde{f}(t)|^2}dt} + \frac{b^4}{4}\mu^2_{K,2}R(g^{\prime \prime}),$$
where $\mu_{K,2}=\int_{-\infty}^{\infty}(x^2K(x)dx)$ and $R(g^{\prime \prime})=\int_{-\infty}^{\infty}{g^{\prime \prime 2}(x)dx}$. Since $R(g^{\prime \prime})$ is unknown we have to estimate it. Silverman (1986)~\cite{BWS} proposed the use a normal reference to estimate it in the error free case, i.e., assume $g(\cdot)$ is normal to estimate $R(g^{\prime \prime})$. This can also be used here and an estimate for $R(g^{\prime \prime})$ is given by $\hat{R}(g^{\prime \prime})=0.375\hat{Var(X)}^{-5/2}\pi^{-1/2}$.

\bigskip

The expression of $AIMSE(b)$ using error from Gamma family and kernel $K(\cdot)$ is given by the following equation,

$$ AIMSE(b)=\frac{1}{\pi n \eta}\int_{0}^{\tan^{-1}(\eta/b)}{\dfrac{(1-\frac{b^2}{\eta^2}\tan^2(\theta))^6}{((\cos(\theta)^{\vartheta+1})\cos(\vartheta\theta))^2}d\theta} + 11520 \frac{b^4}{4} R(g^{\prime\prime}).$$

Minimizing the above expression w.r.t. $b$, we find a bandwidth $b^{\star}$ for kernel estimation. Estimated density function is given by the following expression.
$$ \hat{g}(x)=\frac{1}{\eta \pi n}\sum_{j=1}^n{\int_{0}^{\tan^{-1}(\eta/b^{\star})}{\cos{(\frac{tan(\theta)(x-Z_j)}{\eta})}\dfrac{(1-\frac{b^{\star_2}}{\eta^2}\tan^2(\theta))^3}{((\cos(\theta)^{\vartheta+2})\cos(\vartheta\theta))}d\theta}} $$
and once we get an estimate of $g(x)$ we can integrate it to get an estimate of the c.d.f. function $G(\cdot)$. Here, the integration can be done numerically using any statistical software.
\section{Choosing an optimal Gamma density for Obfuscation}

While choosing an optimal density for obfuscation, one needs to consider at first, the protection of the data from any possible intruder. It is also essential to check how the protected data can be used for statistical inference.

To look, at first, at the aspect of data protection, let us consider a data  set $\{X_1,X_2,\cdots , X_n\}$, having unknown density $g$, convoluted with an error having density function $f$.
The probability $M_{z,\epsilon}^{(f)}$ given by Equation~\eqref{Eqn:M} is unknown for unknown $g$. But using the data it can be estimated by it's empirical form given by Equation \eqref{Eqn:Meas_stat}.

If the error is chosen from the Gamma family, i.e., $f$ can be chosen according to Equation~\eqref{Eqn:GammaFamily}, $M_{z,\epsilon}^{(f)}$ becomes a function of $(\vartheta, \eta)$ and can be estimated by the following statistic.

$$ \hat{M}_{(z,\epsilon)}^{(\vartheta,\eta)} = \begin{cases}
\dfrac{\sum_{i=1}^n{f_{(\vartheta,\eta)}(z-X_i)}\mathbb{I}_{[z-\hat{\sigma}\epsilon < X_i < z+\hat{\sigma}\epsilon]}}{\sum_{i=1}^n{f_{(\vartheta,\eta)}(z-X_i)}} & \mbox{ if $z \in \mathbb{R} \backslash \{X_1,X_2,\cdots , X_n\}$} \\
\mbox{Undefined} & \mbox{ if $z \in \{X_1,X_2,\cdots , X_n\}$}
\end{cases}
$$

We note that the fact that $\hat{M}_{(z,\epsilon)}^{(\vartheta,\eta)}$ is undefined for some values of $z$ is not a problem here. This is because the probability of the convoluted data falling into any finite set of points is 0; $Z$ being absolutely continuous. Now for each fixed $z$ and $\epsilon$ we get a measure estimate $M_{z,\epsilon}^{(\vartheta,\eta)}$. Note that as $z \rightarrow -\infty$ or $\infty$, $M_{z,\epsilon}^{(\vartheta,\eta)} \rightarrow 0$ and the value of $M_{z,\epsilon}^{(\vartheta,\eta)}$ lies between 0 and 1 for all $z \in \mathbb{R}$. Thus, $M_{z,\epsilon}^{(\vartheta,\eta)}$ is a bounded function of $z$ with decreasing tails. Thus, there must exist some $z$ for which it achieves a supremum value. Taking a maximum over the possible range of $z$ we get an estimate of the supremum. 



Note that, $\min_{\epsilon >0} \{ \sup_{z \in \mathbb{R}\backslash \{X_1,X_2,\cdots, X_n\}}\hat{M}_{z,\epsilon}^{(\vartheta,\eta)} \geq \delta \}=Q$ implies,
$$\sup_{z \in \mathbb{R}\backslash \{X_1,X_2,\cdots, X_n\}}\hat{M}_{z,Q}^{(\vartheta,\eta)} \geq \delta$$
and, for $\epsilon<Q$,
$$\sup_{z \in \mathbb{R}\backslash \{X_1,X_2,\cdots, X_n\}}\hat{M}_{z,\epsilon}^{(\vartheta,\eta)} < \delta.$$

For fixed $\vartheta$,$Q$, one can easily observe that,
$$\sup_{z \in \mathbb{R}\backslash \{X_1,X_2,\cdots, X_n\}}\hat{M}_{z,Q}^{(\vartheta,\eta)} = \sup_{z \in \mathbb{R}\backslash \{X_1,X_2,\cdots, X_n\}}\dfrac{\sum_{i=1}^n{f_{(\vartheta,\eta)}(|z-X_i|)}\mathbb{I}_{[|z-X_i| \leq \hat{\sigma}Q]}}{\sum_{i=1}^n{f_{(\vartheta,\eta)}(|z-X_i|)}} $$
is continuous in $\eta$ as $f$ is continuous in $\eta$.
Thus, if we want for a fixed $\delta$, $Q$, $\mu^{(f,\delta)}=Q$, then we can approximate it for fixed $\vartheta$ by finding a corresponding $\eta_\vartheta$ such that the estimate of $\mu^{(f,\delta)}$ for such ($\vartheta,\eta$) is $Q$. This can be done easily using any standard root solving method in a standard statistical software to solve the equation:

$$ \sup_{z \in \mathbb{R}\backslash \{X_1,X_2,\cdots, X_n\}}\hat{M}_{z,Q}^{(\vartheta,\eta)} = \delta $$

Now, to look into the statistical usefulness of the data after obfuscation, we consider the error in estimation of true density curve. To compare the estimation error due to the choice of $f(\cdot)$, we use the Mean Integrated Squared Error ($MISE$) for estimating $g(\cdot)$ for a given choice of Kernel $K$ and bandwidth $b$. The expression for the same was discussed by Meister (2009)~\cite{AM} and is given by $E[\int(\hat{g}(x)-g(x))^2dx]$ which can be written as, 
\begin{eqnarray*}
MISE(g,\hat{g}) &=& \frac{1}{2\pi n}\left[\int_{-\infty}^{\infty}{\frac{|\tilde{K}(tb)|^2}{|\tilde{f}(t)|^2}dt} -\int_{-\infty}^{\infty}|\tilde{K}(tb)|^2|\tilde{g}(t)|^2dt\right]\\
 &&+ \frac{1}{2\pi}\int_{-\infty}^{\infty}{|\tilde{K}(tb)-1|^2 |\tilde{g}(t)|^2 dt}
\end{eqnarray*}
It is easy to observe that $MISE$ depends on $f(\cdot)$ through the term $\int_{-\infty}^{\infty}{\frac{|\tilde{K}(tb)|^2}{|\tilde{f}(t)|^2}dt}$. Thus, minimizing $MISE$ is equivalent to minimizing $\int_{-\infty}^{\infty}{|\frac{\tilde{K}(tb)|^2}{|\tilde{f}(t)|^2}dt}$ with respect to $f(\cdot)$ and for our problem it reduces to minimizing the term $\int_{-\infty}^{\infty}{\frac{|\tilde{K}(tb)|^2}{|\tilde{f}_{(\vartheta,\eta(\vartheta))}(t)|^2}dt}$ with respect to $\vartheta$. $\eta(\vartheta)$ can be chosen to be the minimum scale parameter value that satisfies $\mu^{(\vartheta)}_\delta <Q$ for some preassigned $Q>0$ and $0<\delta<1$.

The problem can thus be viewed as a two parameter graph over $0<\vartheta<1$ and $\eta>0$ such that any ($\vartheta,\eta$) pair above the graph can be used for obfuscation. The idea is to choose an ideal pair ($\vartheta^{\star},\eta^{\star}$) for obfuscation which minimizes the $MISE$ among these eligible pairs. To get a clear view of the idea, one may look at the graphs in Figure 2. The region above the curve in each graph is the region of pairs of $(\vartheta,\eta)$ which can be used for obfuscation and the dot indicates the point ($\vartheta^{\star},\eta^{\star}$) which minimizes the MISE among all the points in this region.

\section{Simulation Results}

\begin{figure}
\centering
\includegraphics[scale=0.5]{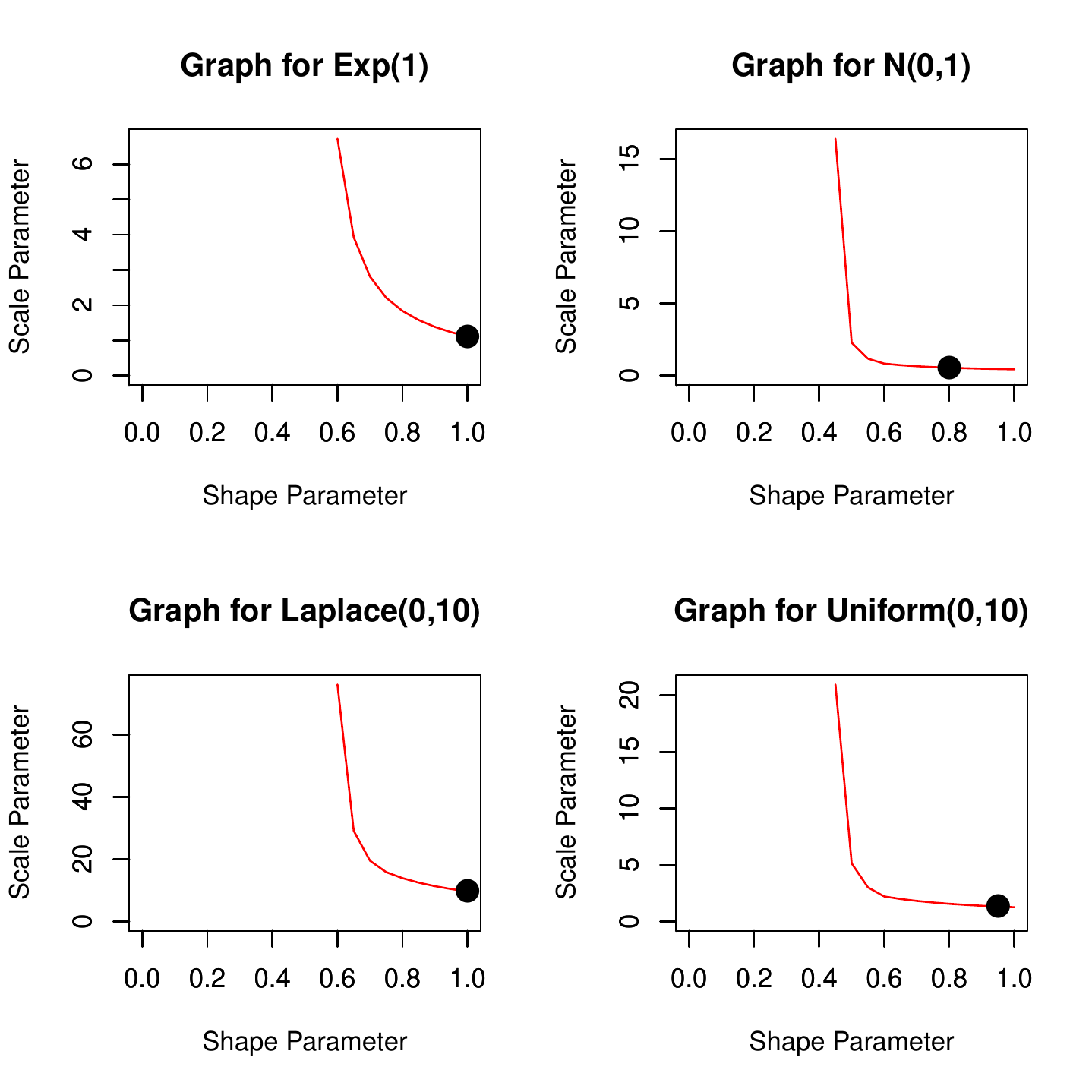}
\caption{Figure showing optimal gamma parameter pair for four different simulated data-sets}
\label{fig:pair}
\end{figure}

\begin{figure}
\centering
r\includegraphics[scale=0.3]{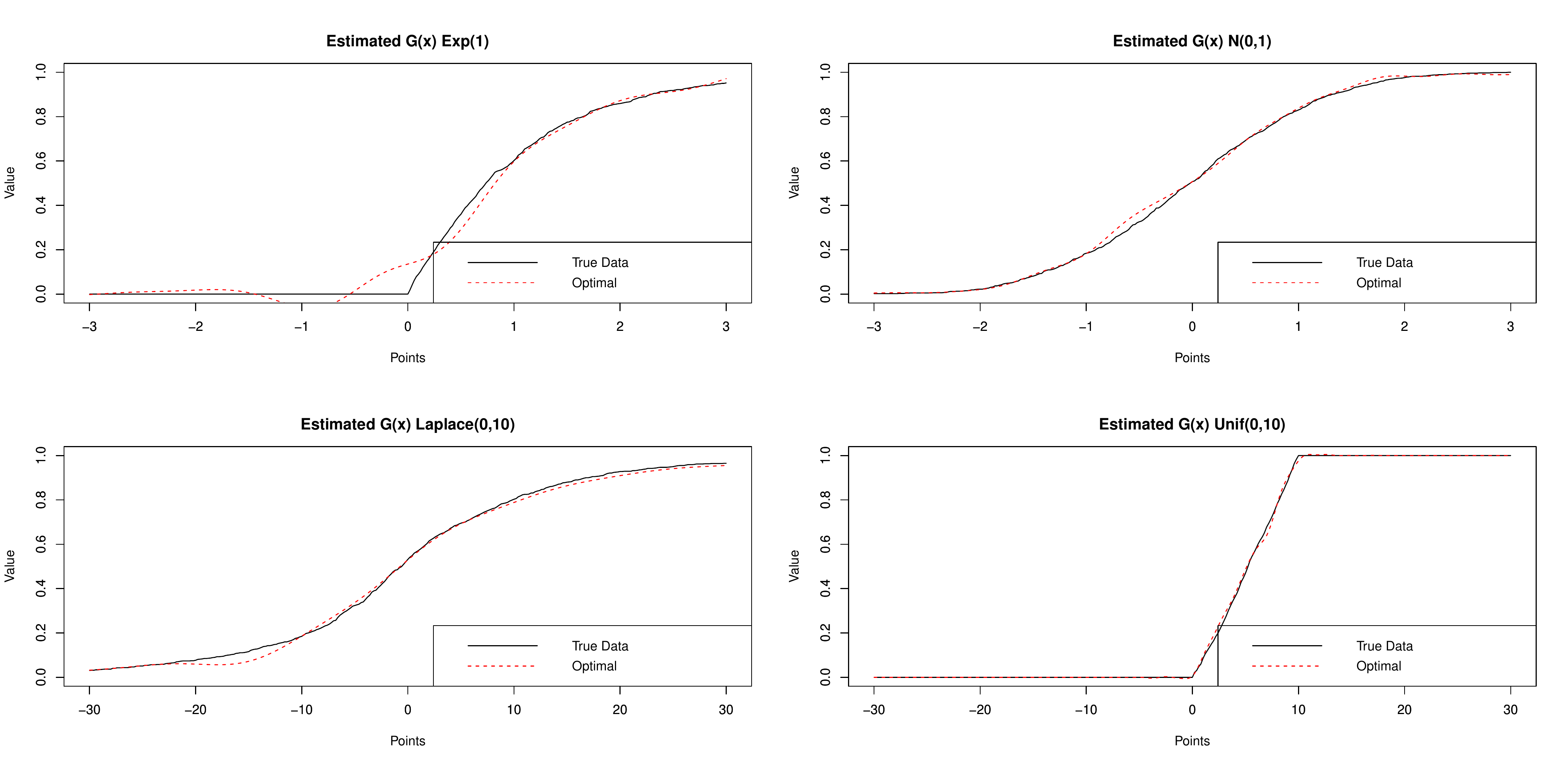}
\caption{Figure showing estimated distribution curves with optimal gamma pair for four different simulated data-sets}
\label{fig:Gx}
\end{figure}

To illustrate the results, here, we simulate a sample of size $n=1000$ from four common distributions to check if the process works. The distributions we chose are Exponential with mean $1$, a standard normal data, a Laplace data with scale $10$ and mean $0$ and a $Uniform(0,10)$ data. The Gamma density parameter pairs ($\vartheta,\eta$) that are eligible for obfuscation, using the discussed method keeping $Q=0.75$ and $\delta=0.9$, are the points that lie above the graph in Figure \ref{fig:pair} . The optimal pair, i.e., the pair that minimizes the $MISE$ (or, the term $\int_{-\infty}^{\infty}{\frac{|\tilde{K}(tb)|^2}{|\tilde{f}_{(\vartheta,\eta(\vartheta))}(t)|^2}dt}$ with respect to $\vartheta$ as mentioned in Section 3) is highlighted with a dot.

One may look at the graphs in Figure \ref{fig:pair}, and note that the optimal point can sometimes be different from being Laplace ( i.e., $\vartheta=1$). In the four cases, considered here, the optimal shape parameters were respectively $(1, 0.8, 1, 0.95)$ with their corresponding required scale parameters $(1.11, 0.548, 9.624, 1.314)$. The graphs of the estimated distribution curves are given in Figure \ref{fig:Gx} .

However, Figure \ref{fig:Gx} do not reveal much about the underlying errors in estimation. Since, there is no existing theory to calculate the standard error in estimation due to noise addition, we try to get a Monte-Carlo estimate of the same. We calculate an approximation to the Mean Squared Error between the true and estimated distribution curve as a measure of error associated in estimation, i.e., given an estimate $\hat{G}(x)$ of the true distribution curve $G(x)$, the error is given by $\frac{1}{N}\sum_{j=1}^N{(\hat{G}(x_j)-G(x_j))^2}$, where $x_j$'s are equidistant points over the range of $G$ with $x_1<x_2< \cdots < x_N$. Since, the range of $G$ is $(-\infty,\infty)$ in case of $N(0,1)$, we used the 6-$\sigma$ range, i.e., $(-3,3)$ which has confidence more than $(1-10^-8)$. We iterated the process $500$ times and calculated the Monte-Carlo estimates of the errors.

One important question that remains unanswered here is how to compare the results of estimation error due to the different noise additions. To do that, let $S_e$ denote the error due to sampling, i.e., the error in estimation of $G(x)$ when there is no error added to the data. Let $L_e$ and $O_e$ respectively denote the error in estimation when Laplace noise is added and noise with optimal parameters is added. The four simulation cases we studied gave us an optimal parameter different from Laplace in two cases. One, when the sample was from $N(0,1)$, and when it was from $U(0,10)$. We calculated the errors in estimation in these cases to study the difference in errors. The following table gives us the result of our study.
\begin{flushleft}
\begin{small}
\begin{tabular}{|c|c|c|c|c|c|c|}
\hline
True Distribution & $S_e$ & $L_e$ & $O_e$ & $\frac{L_e-S_e}{L_e}$ & $\frac{O_e-S_e}{O_e}$ &  $\frac{L_e-S_e}{O_e-S_e}$\\
\hline
$N(0,1)$ & $4.455999.10^{-05}$ & 0.0002517 & 0.0002136 & 0.8229644 & 0.7914206 & 1.225137 \\
\hline
$Uniform(0,10)$ & $3.466508.10^{-05}$ &0.0001798 & 0.0001622 & 0.8071546 & 0.786292 & 1.137586 \\
\hline
\end{tabular}
\end{small}
\end{flushleft}
\medskip

\noindent
Note that the fifth and sixth columns in the table give us an idea about the proportion of error explained due to Laplace and Optimal Noise addition. The last column gives us an idea of the ratio of the errors explained due to the two types of noise addition. One can clearly see that in both the cases the optimal density is a better choice than the Laplace density. However, the significance of the improvement still remains a question.

\section{Conclusion}

Although we developed a method of finding an optimal density to obfuscate a given data-set within the Gamma family given a desired amount of protection to the data we still do not know how significant this improvement is over the Laplace Error. Moreover, we do not know if any other density other than the Gamma  family would be suitable for de-convolution. It would be a fruitful improvement to the existing theories if one can theoretically calculate the standard errors in estimation of the distribution curve. However, these questions are still open and requires a good amount of attention in future.

\section*{Appendix: Proofs}

\noindent{\it Proof of Lemma~\ref{whygamma}}
The Fourier transform of $f_{\vartheta,\eta}(\cdot)$, given by Equation~\eqref{Eqn:GammaFamily} is
\begin{eqnarray*}
\tilde{f}_{\vartheta,\eta}(t)&=&\int_{-\infty}^{\infty}{e^{-itx}f_{\vartheta,\eta}(x)dx}
= \frac{1}{2}((1+it\eta)^{-\vartheta}+(1-it\eta)^{-\vartheta})\\
&=& \left|\frac{1}{2}\mathtt{r}^{-\vartheta}(e^{-i \Theta \vartheta}+e^{i \Theta \vartheta}) \right|
=\mathtt{r}^{-\vartheta}\cos(\Theta \vartheta),
\end{eqnarray*}
where $\mathtt{r}=\sqrt{1+t^2\eta^2}$ and $\Theta=\tan^{-1}(t\eta)$.

Suppose $\eta>1$. Then
$$\mathtt{r}^2>1+t^2>\frac12\{(1+t^2)+(1+t^2)\}\ge\frac12(1+t^2+2|t|)=\frac12(1+|t|)^2,$$
and therefore
\begin{eqnarray*}
\mathtt{r}^{-\vartheta}\cos(\Theta \vartheta) &\le&2^{\vartheta/2}(1+|t|)^{-\vartheta}.
\end{eqnarray*}
Further,
$$\mathtt{r}^2<\eta^2+t^2\eta^2<\eta^2(1+t^2+2|t|)=\eta^2(1+|t|)^2,$$
which means
$$\mathtt{r}^{-\vartheta}\cos(\Theta \vartheta) >\eta^{-\vartheta}(1+|t|)^{-\vartheta}\cos\left(\tan^{-1}(t\eta)\vartheta\right),
\ge\eta^{-\vartheta}\cos\left(\frac{\pi\vartheta}{2}\right)(1+|t|)^{-\vartheta},$$
since $\vartheta\le1$ and $-\pi/2\le\tan^{-1}(t\eta)\le\pi/2$.
Therefore, \eqref{Eqn:OrdSmth} holds with $c_1=\eta^{-\vartheta}\cos(\frac{\pi\vartheta}{2})$, $c_2=2^{\vartheta/2}$ and $\tilde{\vartheta}=\vartheta$.

Now suppose $\eta\le1$. Then
$$\mathtt{r}^2\ge\eta^2(1+|t|^2)>\frac{\eta^2}{2}(1+|t|)^2,$$
and so
\begin{eqnarray*}
\mathtt{r}^{-\vartheta}\cos(\Theta \vartheta) &\le&2^{\vartheta/2}\eta^{-\vartheta}(1+|t|)^{-\vartheta}.
\end{eqnarray*}
In addition
$$\mathtt{r}^2<1+t^2<1+t^2+2|t|=(1+|t|)^2,$$
which means
\begin{eqnarray*}
\mathtt{r}^{-\vartheta}\cos(\Theta \vartheta) &>&\cos\left(\frac{\pi\vartheta}{2}\right)(1+|t|)^{-\vartheta}.
\end{eqnarray*}
Therefore, \eqref{Eqn:OrdSmth} holds with $c_1=\cos(\frac{\pi\vartheta}{2})$, $c_2=2^{\vartheta/2}\eta^{-\vartheta}$ and $\tilde{\vartheta}=\vartheta$.

When $\vartheta>1$, the quantity $\cos(\Theta \vartheta)$ becomes zero whenever $t=\tan[\pi/(2\vartheta)]/\eta$. Therefore, the lower bound of \eqref{Eqn:OrdSmth} does not hold.

\bigskip\noindent{\it Proof of Lemma~\ref{Res:NNobfs}}.
Equation~\eqref{Eqn:M} simplifies in the present case as follows.
\begin{eqnarray*}
M_{z,\epsilon}^{(f)}&=&\dfrac{\frac{1}{\sigma_X \sigma_Y{2\pi}}\int_{-\sigma_X \epsilon}^{\sigma_X\epsilon}{e^{-\frac{(z-x)^2}{2\sigma_X^2}-\frac{x^2}{2\sigma_Y^2}dx}}}{\frac{1}{\sqrt{2\pi(\sigma_X^2+\sigma_Y^2)}}e^{-\frac{z^2}{2(\sigma_X^2+\sigma_Y^2)}}}\\
&=&\frac{1}{\sqrt{2\pi}}\sqrt{\frac{1}{\sigma_X^2}+\frac{1}{\sigma_Y^2}} \int_{\sigma_X \epsilon}^{\sigma_Y \epsilon}{e^{-\frac{(x-\frac{z\sigma_Y^2}{\sigma_Y^2+\sigma_X^2})^2}{2/(\frac{1}{\sigma_X^2}+\frac{1}{\sigma_Y^2})}}dx}\\
&=&\Phi \left( \left(\sigma_X \epsilon-\frac{z\sigma_Y^2}{\sigma_Y^2+\sigma_X^2} \right)\sqrt{\frac{1}{\sigma_X^2}+\frac{1}{\sigma_Y^2}} \right) - \Phi \left( \left(-\sigma_X \epsilon-\frac{z\sigma_Y^2}{\sigma_Y^2+\sigma_X^2} \right)\sqrt{\frac{1}{\sigma_X^2}+\frac{1}{\sigma_Y^2}} \right).
\end{eqnarray*}
The last expression has derivative equal to 0 only at $z=0$ and second derivative at $z=0$ equal to $-2\sigma_X \epsilon (\frac{1}{\sigma_X^2}+\frac{1}{\sigma_Y^2})^{3/2} \frac{z^2 \sigma_Y ^2}{\sigma_X^2 (\sigma_X^2+\sigma_Y^2)}<0$. Therefore, the unique maximum is at $z=0$. It follows that
$$ \sup_{z \in \mathbb{R}}M_{z,\epsilon}^{(f)}= \Phi \left(\sqrt{1+\frac{\sigma_X^2}{\sigma_Y^2}} \cdot \epsilon \right)-\Phi \left(-\sqrt{1+\frac{\sigma_X^2}{\sigma_Y^2}} \cdot\epsilon \right)=2 \Phi \left(\sqrt{1+\frac{\sigma_X^2}{\sigma_Y^2}} \cdot \epsilon \right)-1.$$
and that
$$ \mu^{(f,\delta)}=\min\{\epsilon>0:\Phi \left(\sqrt{1+\frac{\sigma_X^2}{\sigma_Y^2}} \cdot \epsilon \right) \geq \frac{1+\delta}{2}\}
=\dfrac{\tau_{\frac{1+\delta}{2}}}{\sqrt{1+\frac{\sigma_X^2}{\sigma_Y^2}}}.$$

\bibliographystyle{plain}
\bibliography{references_all}

\begin{thebibliography}{10}

\bibitem{CNAM}
C.~Andrieu, N.~de~Freitas, A.~Doucet, and M.~Jordan.
\newblock An introduction to mcmc for machine learning.
\newblock {\em Machine Learning}, 50(1):5--43, 2003.

\bibitem{CH88}
R.J. Carroll and P.~Hall.
\newblock Optimal rates of convergence for deconvolving a density.
\newblock {\em Journal of the American Statistical Association}, 83:1184--1186,
  1988.

\bibitem{TDSPR}
T.~Dalenius and S.P. Reiss.
\newblock Data-swapping: A technique for disclosure control.
\newblock {\em Journal of Statistical Planning and Inference}, 6(1):73--85,
  1982.

\bibitem{ADIG}
A.~Delaigle and I.~Gijbels.
\newblock Practical bandwidth selection in deconvolution kernel density
  estimation.
\newblock {\em Computational Statistics and Data Analysis}, 45(2):249--267,
  2004.

\bibitem{CS}
Cynthia Dwork and Adam Smith.
\newblock Differential privacy for statistics: What we know and what we want to
  learn.
\newblock Journal of Privacy and Confidentiality, 2009.

\bibitem{FAN}
J.~Fan.
\newblock Deconvolution with supersmooth distributions.
\newblock 20(2):155--169, 1992.

\bibitem{FNK}
O.~Frank.
\newblock An application of information theory to the problem of statistical
  disclosure.
\newblock {\em Journal of Statistical Planning and Inference}, 2(2):143--152,
  1978.

\bibitem{GR}
D.~Ghatak and B.~Roy.
\newblock Estimation of true quantiles from quantitative data obfuscated with
  additive noise.
\newblock {\em Journal of Official Statistics}, 34(3):671--694, 2018.

\bibitem{JKWD}
J.~Gouweleeuw, P.~Kooimann, L.~Willenberg, and P.P. Dewolf.
\newblock Post randomization for statistical disclosure control: theory and
  implementation.
\newblock {\em Journal of Official Statistics}, 14(4):463--478, 1998.

\bibitem{RH}
Rob Hall.
\newblock New statistical applications for differential privacy.
\newblock 2012.

\bibitem{MS}
J.~Mares and N.~Shlomo.
\newblock Data privacy using an evolutionary algorithm for invariant pram
  matrices.
\newblock {\em Computational Statistics and Data Analysis}, 79:1--13, 2004.

\bibitem{AM}
Alexander Meister.
\newblock {\em Deconvolution Problems in Nonparametric Statistics}.
\newblock Springer-Verlag, Berlin, 2009.

\bibitem{MRA}
R.A. Moore.
\newblock {\em Controlled Data Swapping Techniques for Masking Use Microdata
  Sets}, volume RR96-04 of {\em Statistical Research Division Report Series}.
\newblock US Bureau of the Census, Statistical Research Division, 1996.

\bibitem{KRR}
K.~Muralidhar, R.~Parsa, and R.~Sarathy.
\newblock A general additive data perturbation method for database security.
\newblock {\em Management Science}, 45(10):1399--1415, 1999.

\bibitem{NAS}
T.K. Nayak and S.A. Adeshiyan.
\newblock On invariant post randomization for statistical disclosure control.
\newblock {\em International Statistical Review}, 84(1):26--42, 2015.

\bibitem{TAZ}
T.K. Nayak, S.A. Adeshiyan, and C.~Zhang.
\newblock A concise theory of randomized response techniques for privacy and
  confidentiality protection.
\newblock In Arijit Chaudhuri, Tasos~C. Christofides, and C.R. Rao, editors,
  {\em Handbook of Statistics}, volume~34, pages 273--286. Elsevier, 2016.

\bibitem{TBL}
T.K. Nayak, B.K. Sinha, and L.~Zayatz.
\newblock Statistical properties of multiplicative noise masking for
  confidentiality protection.
\newblock {\em Journal of Official Statistics}, 27(3):527--544, 2011.

\bibitem{GP}
G.~Paass.
\newblock Disclosure risk and disclosure avoidance for microdata.
\newblock {\em Journal Of Business and Economic Statistics}, 6(4):487--500,
  1988.

\bibitem{WSWB}
W.H. Press, S.A. Teukolsky, W.T. Vetterling, and B.P. Flannery.
\newblock {\em Numerical Recipes: The Art of Scientific Computing}.
\newblock Cambridge University Press, third edition, 2007.

\bibitem{DJJ}
D.~Rebollo-Monedero, J.~Forne, and J.~Domingo-Ferrer.
\newblock From t-closeness-like privacy to postrandomization via information
  theory.
\newblock {\em IEEE Transactions on Knowledge and Data Engineering},
  22(11):1623--1636, 2010.

\bibitem{JRSK}
J.P. Reiter and S.K. Kinney.
\newblock Inferentially valid, partially synthetic data: Generating from
  posterior predictive distributions not necessary.
\newblock {\em Journal of Official Statistics}, 28(4):583--590, 2012.

\bibitem{DBR}
D.B. Rubin.
\newblock Discussion: Statistical disclosure limitation.
\newblock {\em Journal of Official Statistics}, 9(2):461--468, 1993.

\bibitem{RK}
R.~Sarathy, K.~Muralidhar, and R.~Parsa.
\newblock Perturbing non-normal confidential attributes: The copula approach.
\newblock {\em Management Science}, 48(12):1613--1627, 2002.

\bibitem{BWS}
B.W. Silverman.
\newblock {\em Density Estimation for Statistics and Data Analysis}.
\newblock Chapman and Hall, London, 1986.

\bibitem{NCS}
N.C. Spruill.
\newblock The confidentiality and analytic usefulness of masked business
  microdata.
\newblock In {\em Proceedings of the Section on Survey Research Methods}, pages
  602--607. American Statistical Association, 1983.

\bibitem{STC}
Leonard~A. Stefanski and Raymond~J. Carroll.
\newblock Deconvolving kernel density estimators.
\newblock volume 21:2, pages 169--184. Statistics, 1990.

\bibitem{TDK}
P.~Tendick.
\newblock Optimal noise addition for preserving confidentiality in multivariate
  data.
\newblock {\em Journal of Statistical Planning and Inference}, 27(3):341--353,
  1991.

\bibitem{WZ}
L.~Wasserman and S.~Zhou.
\newblock A statistical framework for differential privacy.
\newblock volume 105:489, pages 375--389. Journal of the American Statistical
  Association, 2010.

\end{thebibliography}
\end{document}